\documentclass[12pt,reqno]{amsart}

\usepackage{amssymb}

\newtheorem{theorem}{Theorem}[section]

\newtheorem{lemma}[theorem]{Lemma}

\newtheorem{corollary}[theorem]{Corollary}

\numberwithin{equation}{section}

\newcommand {\NN}{{\mathbb N}}

\newcommand {\CC}{{\mathbb C}}
\newcommand {\RR}{{\mathbb R}}
\newcommand {\ZZ}{{\mathbb Z}}

\begin{document}
\baselineskip=15.5pt

\title[No-go Theorem for Nonabelionic Statistics in GLSMs]{A no-go theorem for nonabelionic statistics in
gauged linear sigma-models}

\author[I. Biswas]{Indranil Biswas}

\address{School of Mathematics, Tata Institute of Fundamental
Research, Homi Bhabha Road, Bombay 400005, India}

\email{indranil@math.tifr.res.in}

\author[N. M. Rom\~ao]{Nuno M. Rom\~ao}

\address{Mathematisches Institut,
Georg-August-Universit\"at G\"ottingen,
Bunsenstra\ss e 3--5, 37073 G\"ottingen,
Germany}

\email{nromao@uni-math.gwdg.de}



\date{}

\begin{abstract}
Gauged linear sigma-models at critical coupling on Riemann surfaces yield self-dual field theories, their classical 
vacua being described by the vortex equations. For local models with structure group ${\rm U}(r)$, we give a 
description of the vortex moduli spaces in terms of a fibration over symmetric products of the base surface 
$\Sigma$, which we assume to be compact. Then we show that all these fibrations induce isomorphisms of fundamental groups. 
A consequence is that all the moduli spaces of multivortices in this class of models have abelian fundamental 
groups. We give an interpretation of this fact as a no-go theorem for the realization of nonabelions through the 
ground states of a supersymmetric version (topological via an A-twist) of these gauged sigma-models. This analysis 
is based on a semi-classical approximation of the QFTs via supersymmetric quantum mechanics on their classical 
moduli spaces.
\end{abstract}

\maketitle
\tableofcontents

\section{Introduction}\label{intro}

Gauged sigma-models appear in a wide spectrum of physical contexts ranging from models of fundamental forces of nature in high-energy physics to effective field
theories describing order parameters of correlated electrons. In the case where the base $\Sigma$ and the target $X$ are chosen to be K\"ahler manifolds, there exist self-dual
versions of these models, and the solutions to the corresponding self-duality equations are called vortices. The
moduli spaces of vortices encode rich geometry and topology.
It has long been recognized that the study of these spaces should unlock much information about the corresponding field theories (either their Riemannian or Lorentzian
versions), which are objects of great interest --- for instance, GLSMs (gauged linear sigma models, corresponding to linear target actions) on surfaces are a basic 
ingredient in mirror symmetry~\cite{HV}. However, some basic questions on these field theories with relevance to model-building in physics remain to be answered. 

In this paper, we address one such question: is it possible to use familiar GLSMs to model exotic statistics in 
quantization --- i.e., couple the moduli dynamics with Aharonov--Bohm holonomies to implement anyonic statistics? 
Such issues have recently been gaining prominence because of the potential relevance of nonabelions~\cite{MooRea} in 
models for the fractional quantum Hall effect and the emergent theory of quantum computation, two areas where both 
sigma-models and gauge theories have been used extensively.

One way to implement statistical phases into a quantum-mechanical system is to construct local systems over their 
configuration spaces, and demand that wavefunctions (or waveforms, in supersymmetric extensions of the models) 
couple to the underlying flat connection. In the situation we want to explore, the role of configuration space is 
taken by a moduli space of static stable solitons described by the vortex equations~\cite{JafTau,Bra,GaDEP}. In a 
semi-classical approach to the quantization, one can lift waveforms on the moduli spaces valued in local systems to 
ordinary waveforms in supersymmetric quantum mechanics on appropriate covers of the configuration 
spaces~\cite{BokRomP, BokRomWeg}. This procedure is convenient when one needs to deal with families of local 
systems, which are naturally parametrized by a continuum --- namely, for each homotopy class of matter field 
configurations, parametrized by a representation variety of the fundamental group of a given moduli space. Thus the fundamental 
groups will severely constrain which type of anyonic statistics can be implemented by this quantization scheme. In 
particular, a necessary condition for the sigma-models to give rise to nonabelionic particles is that at least one 
fundamental group corresponding to multivortices (and on a suitable surface $\Sigma$) is nonabelian.

Our final goal in this paper (see Corollary~\ref{nogo}) is to establish that nonabelian statistics is ruled out for an interesting class of GLSMs defined on all compact Riemann surfaces. The models we will examine are the local ${\rm U}(r)$-GLSMs studied by Baptista in~\cite{BatNAV}, for which the moduli spaces are rigorously understood for a certain
range of the parameters --- more precisely, by a degree $d$, the total area of $\Sigma$ and also the vacuum expectation value in the potential term; this will be reviewed
briefly in Section~\ref{review}. One can realize these moduli spaces~\cite{BisRom} as Quot schemes fibred over symmetric products
of the base Riemann surface $\Sigma$. The fibres of this map in the nonabelian case $r\,>\,1$, which parametrize vortex ``internal structures''~\cite{BatNAV}, are described in
Section~\ref{neqr}. We show in Section~\ref{fundgroups} that this fibration map induces isomorphisms of fundamental groups for all ranks $r$, all degrees $d$ and all Riemann surfaces $\Sigma$. In Section~\ref{snogo}, we draw an implication of this result for the semi-classical quantization scheme that we have referred to above, and comment on contrasting results for other gauged sigma-models.

\section{Gauged linear sigma-models and vortex moduli spaces} \label{review}

We will be concerned with $(1+2)$-dimensional sigma-models for fields defined on a connected, compact and oriented
Riemannian surface $(\Sigma,g_\Sigma)$; the
target $X\,=\,{\rm Mat}_{r\times n} \CC$ (where $r,n \in \NN$) of the sigma-model
is the vector space of complex $r\times n$ matrices on which the structure group ${\rm U}(r)$ acts by multiplication on
the left. In this section we start by fixing our conventions, then recall how moduli spaces of vortices play a role in the description of the classical solutions in these models, and how
they can be understood in terms of algebraic geometry.

\subsection{GLSMs and the vortex equations} \label{sec:adiabatic}

For the moment, we do not impose any restriction on the integers $r$ and $n$. Let us fix an ${\rm U}(r)$-invariant inner
product on the Lie algebra $\mathfrak{u}(r)$, which is tantamount to an ${\rm U}(r)$-equivariant isomorphism
$\kappa: \mathfrak{u}(r)^* \longrightarrow \mathfrak{u}(r)$. We consider the canonical K\"ahler structure
$$\omega_X\,=\,\frac{\sqrt{-1}}{2}\sum_{j,k=1}^{r,n}{\rm d}w_{j,k}\wedge{\rm d}\bar w_{j,k}$$ on $X\,\cong\,
\CC^{rn}$, which is preserved by the left-multiplication by ${\rm U}(r)$ matrices. One readily checks that the
moment maps for this action are of the form
$$
\mu_\tau(\mathbf{w})\,=\,-\frac{\sqrt{-1}}{2} \left(\mathbf{w} {\mathbf{w}}^{\dag} - \tau \mathbf{1}_r \right)
$$
where $\mathbf{1}_r$ is the identity matrix, and $\tau \in \RR$ a constant; we shall write $\mu_\tau^\kappa:= \kappa \circ \mu_\tau$.
Note that the Riemannian structure $g_\Sigma$ together with the orientation also determine a K\"ahler
structure $(\Sigma\, ,\omega_\Sigma\, ,j_\Sigma)$ on $\Sigma$.

Let $P\,\longrightarrow\, \Sigma$ be a principal ${\rm U}(r)$-bundle, and let $${\rm pr}_2\,:\,\RR\times \Sigma \,
\longrightarrow\, \Sigma$$ be the projection onto the second factor.
The GLSMs of our interest describe dynamics parametrized by some time interval $I\subset \RR$.
The solutions determined by these data are stationary points $(\widetilde A,u)$ of the Yang--Mills--Higgs functional
\begin{equation}\label{GLSM}
\mathcal{S}(\widetilde A,u)\,:=\,
-\frac{1}{2e^2}\|F_{\widetilde A}\|^2 + \frac{1}{2}\|{\rm d}^{\widetilde A} u \|^2 - \frac{e^2 \xi}{2}
\|\mu_\tau^\kappa \circ u \|^2 \, ,
\end{equation}
where $e$ and $\xi$ are real parameters. We use $\| \cdot \|$ generically to denote $L^2$-norms for differential forms on $I \times \Sigma$ with respect to the Lorentzian
metric ${\rm d}t^2-g_\Sigma$ on $\RR \times \Sigma$, the inner product on $\mathfrak{u}(r)$ associated to $\kappa$
as well as the standard (flat) metric on $X$ determined by $\omega_X$.
The variables in (\ref{GLSM}) are a ${\rm U}(r)$-connection
$$\widetilde A\,=\,A_t{\rm d}t + A(t)$$ on $({\rm pr}_2^*P)|_{I\times\Sigma}$, and a path $u\,:\,I\, \longrightarrow\, C^\infty(P,X)^{{\rm U}(r)}$ of smooth
${\rm U}(r)$-equivariant maps satisfying appropriate boundary conditions. We may as well interpret $u$ as a section of the associated rank $rn$ complex vector bundle $({\rm pr}_2^*P)\times_{{\rm U}(r)}X$, and $\widetilde A$ as a connection on this vector bundle. As usual, $F_{\widetilde A}$ and ${\rm d}^{\widetilde A}$ denote the curvature and covariant derivative determined by the connection $\widetilde A$.

One should regard $A_t$ as a Lagrange multiplier in this problem, since its time derivative does not feature in the integrand of (\ref{GLSM}). The corresponding Euler--Lagrange equation 
is a constraint that enforces the paths $t\,\longmapsto\, (A(t),u(t))$ to be instantaneously $L^2$-orthogonal to the
orbits of the gauge group 
$\mathcal{G}\,:=\, {\rm Aut}_\Sigma(P)$, acting as $$(A(t), u(t)) \,\longmapsto\,
({\rm Ad}_{g(t)}A(t)-{\sqrt{-1}}g(t){\rm d}g(t), g(t)^{-1}u(t))\, .$$ Ultimately, one is only interested in solutions of the Euler--Lagrange equations
for the dynamical fields $A(t)$ and $u(t)$ (i.e., the equations of motion of the GLSM) up to the action of the group of paths in $\mathcal G$, under which (\ref{GLSM}) is manifestly invariant. 

For these sigma-models, the equations of motion are second-order PDEs in three dimensions, and difficult to study. But there is a well-known procedure to approximate slow-moving solutions 
at the self-dual point $\xi\,=\,1$ in the so-called BPS sector, where $${\rm deg}\, P\,=\, [c_1(A(t))]\,=\,d\, [\Sigma] 
\,\in\, H^2(\Sigma;\, \ZZ)\,\cong\, \ZZ$$ is a positive multiple ($d\,\in\, \NN$) of the fundamental class, by
geodesics in a {\em moduli space} of vortices in two dimensions.
The basic idea~\cite{ManSut} is to approximate solutions by paths of stable static solutions of the model, up to the action of $\mathcal G$. So each point on such a path is itself a 
$\mathcal G$-orbit of solutions of the sigma-model --- more precisely, it can be represented by a constant path $(A,u)$ of minimal potential energy. The potential
energy can be read off from (\ref{GLSM}) to be a sum of $L^2$-norms of the various forms pulled back to Cauchy slices $\{t \}\times \Sigma$, and we can write it as
\begin{eqnarray*}
V(A,u)&=&\frac{1}{2}\int_\Sigma \left( \frac{1}{e^2}|F_{A} |^2 + |{\rm d}^{A}u|^2 + e^2 \xi |\mu_\tau^\kappa \circ u|^2\right)\\
&=&\pi \tau d + \frac{(\xi-1)e^2}{2}\int_\Sigma |\mu_\tau^\kappa \circ u|^2 + \int_\Sigma \left( |{\bar\partial}^{A}_{j_\Sigma}u|^2 + \frac{1}{2} \left|\frac{1}{e}F_{A} + e (\mu_\tau^\kappa \circ u)\omega_\Sigma\right|^2\right) ,
\end{eqnarray*}
where $|\cdot |$ denotes pointwise norms with respect to $g_\Sigma$ and the same target data as before, and
 $\bar\partial^{A}_{j_\Sigma}$ is the usual holomorphic structure on the vector bundle $$P\times_{{\rm U}(r)}X
\,\longrightarrow\, \Sigma$$ constructed from the connection $A$ in $P$
and the complex structure $j_\Sigma$. This rearrangement of the squares is known as the ``Bogomol'ny\u\i\ trick'', and it
makes it clear that, for $\xi\,=\,1$, the minima of $V$ (for each topological class with $d\,>\,0$ fixed) are described by the first-order PDEs
\begin{equation} \label{vortex}
\bar \partial^A_{j_\Sigma} u=0,\qquad F_A +e^2 (\mu_\tau^\kappa \circ u) \omega_\Sigma =0
\end{equation}
called the {\em vortex equations}~\cite{JafTau,Bra,GaDEP}. One can run a similar argument for $d<0$ using
the ``anti-vortex equations'' and employing a variant of the
Bogomol'ny\u\i\ trick.

Given $n,r,d\,\in\, \NN$ as above, we define the moduli space of vortices valued in $r\times n$ matrices at degree
$d$ to be the quotient
\begin{equation}\label{mo}
\mathcal{M}_\Sigma(n,r,d)\,:=\, \{ (A,u): \bar \partial^A_{j_\Sigma} u\,=\,0\,=\,
F_A +e^2 (\mu_\tau^\kappa \circ u) \omega_\Sigma, [c_1(A)]\,=\,d[\Sigma]
 \}/\mathcal{G}\, .
\end{equation}
This space is known to be smooth at least for a range of the parameters (as detailed in Section~\ref{sec:npairs}), and
then it acquires a K\"ahler structure; this K\"ahler structure is often
denoted $\omega_{L^2}$. The underlying metric $g_{L^2}$ (see e.g.~\cite{BatL2M}) can also be thought of as being induced
by the kinetic energy part of the functional (\ref{GLSM}). It is
the geodesic flow of $g_{L^2}$ that approximates the field dynamics at low energies, for initial conditions that solves
the linearization of the equations (\ref{vortex}) --- see~\cite{ManSut}, and~\cite{Stu} for an analysis of the case
$\Sigma\,=\,\CC$, $r\,=\,n\,=\,1, d\,=\,2$.
One can summarize this situation by saying that at low energies the $(1+2)$-dimensional sigma-model is well described by a one-dimensional sigma-model whose target is the
Riemannian manifold $(\mathcal{M}_\Sigma(n,r,d),g_{L^2})$.

\subsection{Vortex moduli, $n$-pairs and Quot schemes} \label{sec:npairs}

Recall that the operator ${\bar \partial}_{j_\Sigma}^A$ acting on sections of any vector bundle associated to a
the principal ${\rm U}(r)$-bundle $P
\,\longrightarrow\, \Sigma$ such as $E\,:=\, P\times_{{\rm U}(r)} \CC^r$
or $P\times _{{\rm U}(r)} X \cong E^{\oplus n}$ endows it with the structure of holomorphic vector
bundle. Thus we can regard a solution $u\,=\,s$ of the first equation in (\ref{vortex})
as a so-called {\em $n$-pair} $(E\, ,s)$ with $s \,\in\, H^0(\Sigma,\, E^{\oplus n})\,\cong\, H^0(\Sigma,\,
E)^{\oplus n}$, see~\cite{BDW}. One groups any two of such objects $(E,s)$ and $(E',s')$ in an isomorphism equivalence class whenever there is
an isomorphism of holomorphic vector bundles $\psi\,:\, E\,\longrightarrow\, E'$ with $\psi^* s'\,=\,s$. Alternatively, one says that such equivalence classes correspond to orbits of the action
of the complexification $\mathcal{G}^\CC$ of the gauge group we introduced above, which preserves the first equation in~(\ref{vortex}) but not the second one.

There is a way to relate equivalence classes of $n$-pairs with points in the moduli space~(\ref{mo}) we introduced
above, generalizing results in~\cite{Br2} for $n\,=\,1$. Let ${\rm Vol}(\Sigma)\,:=\,\int_\Sigma \omega_\Sigma$ denote
the total area of the surface. One shows (\cite{BDW,BDGW}) that whenever $(E,s)$ is {\em $\frac{e^2 \tau}{4\pi}{\rm Vol}(\Sigma)$-stable} in the sense that~\cite{BisRom}
\begin{itemize}
\item
$\displaystyle \frac{{\rm deg}\, E'}{{\rm rk}\, E'} \,<\, \frac{e^2 \tau}{4\pi} {\rm Vol}(\Sigma) $ for all
holomorphic subbundles $E'\subseteq E$, and
\item
$\displaystyle \frac{{\rm deg}\,(E/E_s)}{{\rm rk}(E/E_s)} \,>\, \frac{e^2 \tau}{4\pi} {\rm Vol}(\Sigma) $ for all
holomorphic subbundles $E_s \subsetneq E$ containing all the component sections of $s$,
\end{itemize}
there is exactly one $\mathcal G$-orbit of solutions $(A,u)$ to the second equation (\ref{vortex}) inside a
$\mathcal{G}^\CC$-orbit of a solution $(A,u)$ to the first equation, and this produces a bijection between the two
quotients that preserves their natural complex structures (this is an example of
the so-called Hitchin--Kobayashi correspondence). One can check that whenever
$n\,\geq\, r$ and 
\begin{equation}\label{stable}
e^2 \tau\, {\rm Vol}(\Sigma) \,>\, 4 \pi \, {{\deg} (E)}
\end{equation} 
are assumed (as we will do from now on), then both conditions itemized above are automatically met if $s$ has maximal
rank generically on $\Sigma$. The main advantage is that one can then describe $\mathcal{M}_\Sigma(n,r,d)$ purely in terms of algebraic geometry. We point out that
there are other natural stability conditions on $n$-pairs; the reader is referred to~\cite{Tha} for a discussion of this, as well as for an illustration of 
how these moduli spaces may undergo rather dramatic changes for other values of the stability parameter.

In~\cite{BisRom}, setting $n\,\geq\, r$ we considered for each $n$-pair $(E,s)$ as above a homomorphism of
holomorphic vector bundles 
\begin{equation}\label{e2}
f_s \,:\, {\mathcal O}^{\oplus n}_\Sigma \,\longrightarrow\, E
\end{equation}
given by $(x ; c_1 ,\cdots ,c_n)
\,\longmapsto \,\sum_{i=1}^n c_i\cdot s_i(x)$, where $x \,\in\,
\Sigma$ and $c_i \,\in\, \mathbb C$. The image ${\rm im}(f_s)$ is a coherent
analytic sheaf which is torsion-free because it is contained in the
torsion-free sheaf $E$, and ${\rm im}(f_s)$ generically generates $E$.
Considering the dual homomorphism to (\ref{e2}), we obtain a short exact sequence
\begin{equation}\label{ses}
0 \longrightarrow 
E^* \stackrel{f^*_s}{\longrightarrow} (\mathcal{O}_\Sigma^{\oplus n})^*
\,=\,\mathcal{O}_\Sigma^{\oplus n} \longrightarrow \mathcal{Q} \longrightarrow 0\, ,
\end{equation}
where $\mathcal{Q}$ is of rank $n-r$; so $\mathcal{Q}$ a torsion sheaf if $n\,=\, r$.
The support of the torsion part of $\mathcal{Q}$ consists of points points
where ${\rm im}(f_s)$ fails to generate $E$.
Given an auxiliary ample line bundle $\mathcal{L}\,\longrightarrow\, \Sigma$ of sufficiently large degree so that
$$H^1(\Sigma,\,E^*\otimes \mathcal{L}^{\otimes \delta})\,=\,0$$ for all
integers $\delta \,\ge \,\delta_E$ (where $\delta_E \,\in\, \NN$ depends only on $n$, $r$ and
$\deg E$), and after tensoring (\ref{ses}) with $\mathcal{L}^{\otimes \delta}$, one extracts a short exact sequence of vector spaces
\begin{equation}\label{mapQ}
0\longrightarrow H^0(\Sigma,\,E^*\otimes\mathcal{L}^{\otimes \delta}) \longrightarrow
H^0(\Sigma,\,(\mathcal{L}^{\otimes \delta})^{\oplus n}) \stackrel{Q}{\longrightarrow} H^0(\Sigma,\,\mathcal{Q} \otimes \mathcal{L}^{\otimes \delta}) \longrightarrow 0
\end{equation}
from the long exact sequence of cohomologies associated to \eqref{ses} tensored with $\mathcal{L}^{\otimes \delta}$.
In~\cite[Lemma~3.2]{BisRom}, it was shown that the homomorphism $f_s$ (and thus the $n$-pair $(E,s)$ up to
isomorphism) can be reconstructed from
the quotient $Q$ in~(\ref{mapQ}). This construction realizes (the algebraic-geometric version of) the moduli space of vortices $\mathcal{M}_\Sigma(n,r,d)$ as a Quot scheme~\cite{Gr}.
One can take advantage of this viewpoint to study properties of the moduli space under the
assumption (\ref{stable}) --- for example, show that it is smooth, projective, and realize the K\"ahler class $[\omega_{L^2}]$ geometrically~\cite{BisRom}.

\section{Internal structures of nonabelian local vortices} \label{neqr}

In this section we shall assume that $r\, =\, n$, which is usually referred to as the case of {\em local} vortices, in contrast with the nonlocal case $r<n$. We will give a description
of our Quot scheme in the nonabelian situation $r\,>\, 1$, by means of Hecke 
modifications~\cite{FB, HL} on holomorphic vector bundles over $\Sigma$. 
So from now on we shall set
\begin{equation}\label{mosimp}
{\mathcal M}_\Sigma\,:=\, {\mathcal M}_\Sigma(n,n,d)
\end{equation}
(see \eqref{mo}). We know that the objects parametrized by this space can also be described as isomorphism classes of $n$--pairs $(E,s)$ if the
condition \eqref{stable} holds; the corresponding $n$ sections generically generate
the vector bundle $E\,\longrightarrow \,\Sigma$.

Take any $(E\, ,s)\, \in\, {\mathcal M}_\Sigma$. Consider the homomorphism
$f_s$ in \eqref{e2}. Since the sections of $E$ in $s$ generate $E$
generically, we know that the quotient 
${\rm coker}\, f_s = E/f_s({\mathcal O}^{\oplus n}_\Sigma)$
is a torsion sheaf supported on finitely many points, and we have
$$
\dim H^0(\Sigma,\, E/f_s({\mathcal O}^{\oplus n}_\Sigma))\,=\, d\, .
$$
Note that ${\mathcal M}_\Sigma(1,1,d)$ is identified with the $d$-fold symmetric
product ${\rm Sym}^d(\Sigma)$ by sending any $(E\, ,s) \, \in\, {\mathcal M}_\Sigma(1,1,d)$
to the scheme-theoretic support of the quotient $E/f_s({\mathcal O}^{\oplus n}_\Sigma)$.
For general $n\, \geq\, 1$, consider the $n$-th exterior product
$$
\bigwedge\nolimits^n f_s\, :\, \bigwedge\nolimits^n{\mathcal O}^{\oplus n}_\Sigma
\,=\, {\mathcal O}_\Sigma \,\longrightarrow\, \bigwedge\nolimits^n E
$$
of the homomorphism in \eqref{e2}. Let
\begin{equation}\label{Phi}
\Phi\, :\, {\mathcal M}_\Sigma\,=\, {\mathcal M}_\Sigma(n,n,d)
\, \longrightarrow\, {\mathcal M}_\Sigma(1,1,d)\,=\, {\rm Sym}^d(\Sigma)
\end{equation}
be the map that sends any $(E\, ,s)$ to the pair $(\bigwedge\nolimits^nE\, ,
\bigwedge\nolimits^n f_s)$ constructed above from $(E\, ,s)$. To explain what this map
$\Phi$ does, let $m_x$ denote the dimension
of the stalk of $E/f_s({\mathcal O}^{\oplus n}_\Sigma)$ at each
point $x\, \in\, \Sigma$. Since $E/f_s({\mathcal O}^{\oplus n}_\Sigma)$ is a torsion sheaf,
we have $m_x\,=\, 0$ for all but finitely many $x$.
The map $\Phi$ sends $(E\, ,s)$ to $\sum_{x\in\Sigma} m_x\cdot x$.

The map $\Phi$ in \eqref{Phi} is clearly surjective. In what follows,
we shall describe step by step its fibers, which parametrize the vortex {\em internal structures} introduced in~\cite{BatNAV}. We shall obtain a description of the moduli space as
a stratification by the type of the partitions of $d$ associated to effective
divisors of degree $d$. This description will be fully algebraic-geometric, contrasting to the one in reference~\cite{BatNAV}, which depended on a choice of Hermitian inner product
on the fibres of $E$.

\subsection{The case of distinct points}\label{sec2d}

Let ${\mathbb P}^{n-1}$ be the projective space parametrizing all hyperplanes
in ${\mathbb C}^n$. Take $d$ distinct points
$$
x_1\, ,\cdots\, ,x_d \, \in\, \Sigma\, .
$$
Let $\underline{x}\, \in\, {\rm Sym}^d(\Sigma)$ be the point defined
by $\{x_1\, ,\cdots\, ,x_d\}$. We will show that the fiber of $\Phi$ over
$\underline{x}$ is the Cartesian product $({\mathbb P}^{n-1})^d$.
This is a description of the generic fiber of the map $\Phi$, and it 
coincides with the one in~\cite{BatNAV}.

Take any $(H_1\, ,\cdots\, ,H_d)\,\in\, ({\mathbb P}^{n-1})^d$. So each
$H_i$ is a hyperplane in ${\mathbb C}^n$. The fiber of the trivial
vector bundle ${\mathcal O}^{\oplus n}_\Sigma$
over $x_i$ is identified with ${\mathbb C}^n$. Thus the hyperplane
$H_i\, \subset\, {\mathbb C}^n$ defines a hyperplane
$\widetilde{H}_i$ in the fiber of ${\mathcal O}^{\oplus n}_\Sigma$
over the point $x_i$. Let
\begin{equation}\label{i1}
\widetilde{q}\,:\, {\mathcal O}^{\oplus n}_\Sigma \,\longrightarrow
\, \bigoplus_{i=1}^d ({\mathcal O}^{\oplus n}_\Sigma)_{x_i}/
\widetilde{H}_i
\end{equation}
be the quotient map. The kernel of $\widetilde{q}$ will be denoted
by $\widetilde{\mathcal K}$, and we have the following short exact sequence
of sheaves on $\Sigma$:
$$
0\,\longrightarrow\, \widetilde{\mathcal K}\, \stackrel{h}{\longrightarrow}\,
{\mathcal O}^{\oplus n}_\Sigma\, \stackrel{\widetilde{q}}{\longrightarrow}
\, \bigoplus_{i=1}^d ({\mathcal O}^{\oplus n}_\Sigma)_{x_i}/
\widetilde{H}_i\, \longrightarrow\, 0\, .
$$
Now consider the dual of the homomorphism $h$ above,
$$
h^*\, :\, ({\mathcal O}^{\oplus n}_\Sigma)^*\,=\,{\mathcal O}^{\oplus n}_\Sigma
 \, \longrightarrow\,\widetilde{\mathcal K}^*\, .
$$
It is easy to see that the pair $(\widetilde{\mathcal K}^*\, ,h^*)$
defines a point in the fiber of $\Phi$ (see \eqref{Phi})
over the point $\underline{x}$, and that each point in the fiber can be obtained 
by choosing the hyperplanes $H_i$ suitably. This construction identifies
the fiber of $\Phi$ over
$\underline{x}$ with the Cartesian product $({\mathbb P}^{n-1})^d$.

Employing the usual terminology, we can say that we have constructed the bundle 
$E\,=\,\widetilde{\mathcal K}$
of an $n$-pair by performing $d$ elementary Hecke modifications (one at each $x_i$)
on the trivial bundle of rank $n$ over $\Sigma$, and the inclusion $h$ yields the 
morphism $h^*\,=\, f_s$ in (\ref{e2}) which is equivalent to a 
holomorphic section $s\,\in\, H^0 (\Sigma\, ,E^{\oplus n})$ that generate $E$
over a nonempty Zariski open subset of $\Sigma$.

\subsection{Case of multiplicity two}

Now take $d-1$ distinct points
$$
x_1\, ,\cdots\, ,x_{d-1}\, \in\, \Sigma\, .
$$
Let $\underline{x}\, \in\, {\rm Sym}^d(\Sigma)$ be the point defined
by $2x_1+\sum_{j=2}^{d-1} x_j$. We will describe the fiber
of $\Phi$ over $\underline{x}$.

Let $H_1$ be a hyperplane in ${\mathbb C}^n$. Let
$$
q_1\, :\, {\mathcal O}^{\oplus n}_\Sigma \,\longrightarrow
\, ({\mathcal O}^{\oplus n}_\Sigma)_{x_1}/\widetilde{H}_1
$$
be the quotient map, where, just as in \eqref{i1}, $\widetilde{H}_1$ 
is the hyperplane in the fiber of ${\mathcal O}^{\oplus n}_\Sigma$
over $x_1$ given by $H_1$. Let ${\mathcal K}(H_1)$ denote the
kernel of $q_1$. So we have a short exact sequence of sheaves
on $\Sigma$
\begin{equation}\label{i2}
0\,\longrightarrow\, {\mathcal K}(H_1)\,\stackrel{h'}{\longrightarrow}\,
{\mathcal O}^{\oplus n}_\Sigma\,\longrightarrow\,({\mathcal O}^{\oplus
n}_\Sigma)_{x_1}/\widetilde{H}_1\,\longrightarrow\, 0\, .
\end{equation}

Consider the space ${\mathcal S}_2$ of all objects of the form
$$
(H_1\, , H_2\, , \cdots\, , H_{d-1}\, ; H^1)\, ,
$$
where $H_i$, $1\,\leq\, i\, \leq\, d-1$, is a hyperplane
in ${\mathbb C}^n$, and $H^1$ is a hyperplane in the fiber over $x_1$ of the above
vector bundle ${\mathcal K}(H_1)$. There is a natural surjective map from this 
space ${\mathcal S}_2$ to the fiber of $\Phi$ over the point 
$\underline{x}$ of ${\rm Sym}^d(\Sigma)$. To construct this map, first 
note that for any point $x\, \in\, 
\Sigma$ different from $x_1$, the fibers of ${\mathcal K}(H_1)$ and
${\mathcal O}^{\oplus n}_\Sigma$ over $x$ are identified using the homomorphism
$h'$ in \eqref{i2}. Hence for any $2\,\leq\, j\, \leq\, d-1$, the hyperplane $H_j$ gives
a hyperplane in the fiber of ${\mathcal K}(H_1)$ over the point
$x_j$; this hyperplane in the fiber of ${\mathcal K}(H_1)$ will be denoted
by $\widetilde{H}_j$.
Let $\mathcal K$ be the holomorphic vector bundle over $\Sigma$ 
that fits in the following short exact sequence of sheaves:
\begin{equation}\label{i3}
0\,\longrightarrow\, {\mathcal K}\, \stackrel{h}{\longrightarrow}\,
{\mathcal K}(H_1) \, \longrightarrow\, ({\mathcal K}(H_1)_{x_1}/H^1)
\oplus \bigoplus_{j=2}^{d-1} {\mathcal K}(H_1)_{x_j}/\widetilde{H}_j
\,\longrightarrow\, 0\, .
\end{equation}
Consider the composition
$$
h'\circ h\, :\, {\mathcal K}\,\longrightarrow\,{\mathcal O}^{\oplus n}_\Sigma\, ,
$$
where $h'$ and $h$ are constructed in \eqref{i2} and \eqref{i3} respectively. Let
$$
(h'\circ h)^*\, :\, ({\mathcal O}^{\oplus n}_\Sigma)^*\,=\,{\mathcal O}^{\oplus n}_\Sigma
\,\longrightarrow\, {\mathcal K}^*
$$
be its dual. The pair $({\mathcal K}^*\, ,(h'\circ h)^*)$ defines an element of the moduli
space ${\mathcal M}_\Sigma$ that lies over $\underline{x}$ for
the surjection $\Phi$. Moreover, each element in the fiber over $\underline{x}$
arises in this way for some element of ${\mathcal S}_2$.

Sending any $(H_1\, , H_2\, , \cdots\, , H_{d-1}\, ; H^1)\,\in\,
{\mathcal S}_2$ to $(H_1\, , H_2\, , \cdots\, , H_{d-1})\,\in\,({\mathbb P}^{n-1})^{d-1}$,
we see that ${\mathcal S}_2$ is a projective bundle over $({\mathbb P}^{n-1})^{d-1}$ of
relative dimension $n-1$. Therefore, we have the following lemma:

\begin{lemma}\label{lem-f2}
The fiber $\Phi^{-1}(\underline{x})$ admits a natural isomorphism
with ${\mathcal S}_2$. The variety ${\mathcal S}_2$ is
a projective bundle over $({\mathbb P}^{n-1})^{d-1}$ of relative dimension $n-1$.
\end{lemma}

\subsection{Case of multiplicity $m>2$}

Let $m$ be an integer satisfying $2 < m \le d$, and 
fix $d-m+1$ distinct points
$x_1\, ,x_2\, ,\cdots\, , x_{d-m+1}$ of $\Sigma$. Let
$$
\underline{x}\, \in\, {\rm Sym}^d(\Sigma)
$$
be the point defined by $m\cdot x_1+\sum_{j=2}^{d-m+1}x_j$.

Let $H_1$ be a hyperplane in ${\mathbb C}^n$. Construct
${\mathcal K}(H_1)$ as in \eqref{i2}. Let
$$
H^1\, \subset\, {\mathcal K}(H_1)_{x_1}
$$
be a hyperplane in the fiber of the vector bundle
${\mathcal K}(H_1)$ over the point $x_1$. Let
${\mathcal K}(H^1)$ be the holomorphic vector bundle
over $\Sigma$ that fits in the following exact sequence
of sheaves
$$
0\,\longrightarrow\, {\mathcal K}(H^1)\,\longrightarrow\,
{\mathcal K}(H_1)\,\longrightarrow\,
{\mathcal K}(H_1)_{x_1}/H^1\,\longrightarrow\, 0\, .
$$
Now fix a hyperplane
$$
H^2\, \subset\, {\mathcal K}(H^1)_{x_1}
$$
in the fiber of ${\mathcal K}(H^1)$ over $x_1$. Let
${\mathcal K}(H^2)$ be the holomorphic vector bundle over
$\Sigma$ that fits in the following short exact sequence
of sheaves
$$
0\,\longrightarrow\, {\mathcal K}(H^2)\,\longrightarrow\,
{\mathcal K}(H^1)\,\longrightarrow\,
{\mathcal K}(H^1)_{x_1}/H^2\,\longrightarrow\, 0\, .
$$

Inductively, after $j$ steps performed as above, fix a hyperplane
$$
H^{j+1} \, \subset\, {\mathcal K}(H^j)_{x_1}
$$
and construct the vector bundle ${\mathcal K}(H^{j+1})$
that fits in the short exact sequence
\begin{equation}\label{i4}
0\,\longrightarrow\, {\mathcal K}(H^{j+1})\,\longrightarrow\,
{\mathcal K}(H^j)\,\longrightarrow\,
{\mathcal K}(H^j)_{x_1}/H^{j+1}\,\longrightarrow\, 0\, .
\end{equation}

Consider the space ${\mathcal S}_m$ of all elements of the form
$$
(H_1\, , H_2\, , \cdots\, , H_{d-m+1}\, ; 
H^1\, , H^2\, ,\cdots\, , H^{m-1})\, ,
$$
where $H_i$ is a hyperplane in ${\mathbb C}^n$,
while $H^1$ is a hyperplane in ${\mathcal K}(H_1)_{x_1}$,
and each $H^j$ is a hyperplane in the fiber over $x_1$ of
the vector bundle ${\mathcal K}(H^{j-1})$. There is a natural
map from ${\mathcal S}_m$ to the fiber of $\Phi$ over the
point $\underline{x}$. To construct this map, first
note that from \eqref{i4} it follows inductively that for any point
$x\,\in\, \Sigma\setminus\{x_1\}$, the fiber of ${\mathcal K}(H^{j+1})$
over $x$ is identified with the fiber of 
${\mathcal O}^{\oplus n}_\Sigma$ over $x$. Therefore,
for any $2 \,\le\, i \,\le\, d-m+1$, the hyperplane
$H_i$ in ${\mathbb C}^n$ defines a hyperplane in the
fiber of ${\mathcal K}(H^{m-1})$ over the point $x_i$;
this hyperplane in the
fiber ${\mathcal K}(H^{m-1})_{x_i}$ will be denoted
by $\widetilde{H}_i$. Let $\mathcal K$ be the holomorphic
vector bundle over $\Sigma$ that fits in the following
short exact sequence of sheaves:
\begin{equation}\label{i5}
0\,\longrightarrow\, {\mathcal K}\, \stackrel{h}{\longrightarrow}\,
{\mathcal K}(H^{m-1}) \, \longrightarrow\,
\bigoplus_{j=2}^{d-m+1} {\mathcal K}(H^{m-1})_{x_j}/\widetilde{H}_j
\,\longrightarrow\, 0\, .
\end{equation}
Let $h'\, :\, {\mathcal K}(H^{m-1}) \, \longrightarrow\,
{\mathcal O}^{\oplus n}_\Sigma$ be the natural inclusion.

The pair $({\mathcal K}^*\, ,(h'\circ h)^*)$ in \eqref{i5} defines a
point of the moduli
space ${\mathcal M}_\Sigma$ that lies over $\underline{x}$. Every point in the fiber over
$\underline{x}$ arises in this way for some element of ${\mathcal S}_m$.

Consider the $m-1$ maps
$$
{\mathcal S}\,\longrightarrow\, \cdots \,\longrightarrow\,
({\mathbb P}^{n-1})^{d-m+1}
$$
defined by
$$
(H_1\, , H_2\, , \cdots\, , H_{d-m+1}\, ;
H^1\, , H^2\, ,\cdots\, , H^{m-1})\,\longmapsto
$$
$$
(H_1\, , H_2\, , \cdots\, , H_{d-m+1}\, ;
H^1\, , H^2\, ,\cdots\, , H^{m-2})\longmapsto \cdots \longmapsto
(H_1\, , H_2\, , \cdots\, , H_{d-m+1})\, .
$$
Each of these is a projective bundle of relative dimension $n-1$.
Therefore, we have the following generalization of Lemma
\ref{lem-f2}:

\begin{lemma}\label{lem-f3}
The fiber $\Phi^{-1}(\underline{x})$ is naturally isomorphic to
${\mathcal S}_m$. There is a chain of $m-1$ maps starting from
${\mathcal S}_m$ ending in $({\mathbb P}^{n-1})^{d-m+1}$ such that
each one is a projective bundle of relative dimension $n-1$.
\end{lemma}

\subsection{The general case}

The general case is not harder to understand than the previous case.

Take any point $\underline{x}\, :=\, \sum_{i=1}^a m_i\cdot x_i$ of 
${\rm Sym}^d(\Sigma)$, where $m_i$ are arbitrary positive integers
adding up to $d$ and $x_i \,\in\, \Sigma$, $i\,=\,1,\cdots, a$.
For each point $x_i$, fix data $(H_i\, , H^1_i\, ,\cdots\, , H^{m_i-1}_i)$,
where $H_i$ is a hyperplane in ${\mathbb C}^n$, and the $H^j_i$
are hyperplanes in the fibers, over $x_i$, of vector bundles
constructed inductively as in the previous case. From the set
of such objects, there is a
canonical isomorphism to the fiber of $\Phi$ over $\underline{x}$. Indeed,
this is obtained by repeating the above argument.

\section{Fundamental groups of nonabelian vortex moduli spaces}\label{fundgroups}

In this section, we take advantage of the map $\Phi$ defined in the
previous section to compute the fundamental group $\pi_1 (\mathcal{M}_\Sigma)$.

\begin{theorem}\label{lefg}
The homomorphism $\Phi_*\, :\, \pi_1({\mathcal M}_\Sigma)
\, \longrightarrow\, \pi_1({\rm Sym}^d(\Sigma))$ induced by $\Phi$ in \eqref{Phi}
is an isomorphism.
\end{theorem}

\begin{proof}
Let
$$
\mathcal{D}\, \subset\, \Sigma\times \text{Sym}^{d}(\Sigma)
$$
be the universal divisor, consisting of all $(x\, , \underline{ y}
\,=\, \sum_{i=1}^a m_i\cdot y_i )$ such that $\sum_{i=1}^a m_i = d$ and 
$x\, \in\, \{y_1\, ,\cdots\, , y_a\}$. Then the natural homomorphism
$$
{\mathcal O}^{\oplus n}_{\Sigma\times \text{Sym}^{d}(\Sigma)}
\,\hookrightarrow\,
{\mathcal O}_{\Sigma\times \text{Sym}^{d}(\Sigma)}(\mathcal{D})\oplus
{\mathcal O}^{\oplus (n-1)}_{\Sigma\times \text{Sym}^{d}(\Sigma)}
$$
over $\Sigma\times \text{Sym}^{d}(\Sigma)$ produces a morphism
\begin{equation}\label{theta}
\theta\, :\, \text{Sym}^{d}(\Sigma)\, \longrightarrow\, {\mathcal M}_\Sigma\, .
\end{equation}
This is a section of $\Phi$ in the sense that
\begin{equation}\label{e11}
\Phi\circ\theta\,=\, \text{Id}_{\text{Sym}^d(\Sigma)}\, .
\end{equation}
Therefore, the homomorphism $\Phi_*$ in the statement of the lemma is surjective.

Let $U\, \subset\, \text{Sym}^{d}(X)$
be the Zariski open subset parametrizing reduced effective divisors in $\Sigma$, i.e., points of the form $\underline{y}=\sum_{i=1}^d y_i$ with all $y_i$ distinct.
Let
\begin{equation}\label{ep}
\theta_0\, :=\, \theta\vert_U :\, U\,\longrightarrow\, {\mathcal M}_\Sigma
\end{equation}
be the restriction of the map $\theta$ in \eqref{theta}. Also, consider the restriction
\begin{equation}\label{vpp}
\Phi_0\, :=\, \Phi\vert_{\Phi^{-1}(U)}\, :\, \Phi^{-1}(U)\,\longrightarrow\, U\, .
\end{equation}
As we saw in Section \ref{sec2d}, the fibers of $\Phi_0$ are identified with
$({\mathbb P}^{n-1})^{d}$. From the
homotopy exact sequence associated to $\Phi_0$ it now follows that the induced
homomorphism of fundamental groups
$$
\Phi_{0,*}\, :\, \pi_1(\Phi^{-1}(U))\,\longrightarrow\, \pi_1(U)
$$
is an isomorphism. The variety ${\mathcal M}_\Sigma$ is smooth, and $\Phi^{-1}(U)$ is a
nonempty Zariski open subset of it. Therefore, the homomorphism
$$
\iota_*\, :\, \pi_1(\Phi^{-1}(U))\,\longrightarrow\, \pi_1({\mathcal M}_\Sigma)
$$
induced by the inclusion $\iota\, :\, \Phi^{-1}(U)\,\hookrightarrow\,
{\mathcal M}_\Sigma$ is surjective. Since $\Phi_{0,*}$ is an isomorphism, this
implies that the homomorphism
$$
\theta_{0,*}\, :\, \pi_1(U)\,\longrightarrow\, \pi_1({\mathcal M}_\Sigma)
$$
induced in $\theta_0$ in \eqref{ep} is surjective. Since $\theta_0$ extends to $\theta$,
this immediately implies that the homomorphism
$$
\theta_*\, :\, \pi_1(\text{Sym}^{d}(X))\,\longrightarrow\, \pi_1({\mathcal M}_\Sigma)
$$
induced in $\theta$ in \eqref{theta} is surjective. Since $\theta_*$ is surjective, and
the composition $\Phi_*\circ\theta_*$ is injective (see \eqref{e11}) we
conclude that $\Phi_*$ is injective.
\end{proof}

\begin{corollary}\label{pi1M}
$\pi_1(\mathcal{M}_\Sigma(n,n,d)) \cong H_1(\Sigma;\ZZ)$ for all $n\in \NN$ and $d>1$. 
\begin{proof}
If $n>1$, one has $\pi_1(\mathcal{M}_\Sigma(n,n,d))\cong \pi_1({\rm Sym}^d(\Sigma))$ according to Theorem~\ref{lefg}. The same is true in the 
abelian case $n=1$; this follows directly from $\mathcal{M}_\Sigma(1,1,d)\cong {\rm Sym}^d(\Sigma)$ under the stability assumption~(\ref{stable}), which was established
e.g. in~\cite{Bra,GaDEP}. Finally, 
by the Dold--Thom theorem~\cite{DolTho} one also has $\pi_1({\rm Sym}^d(\Sigma))\cong H_1(\Sigma;\ZZ)$ for $d>1$.
\end{proof}
\end{corollary}

\section{No-go theorem for nonabelions in gauged linear sigma-models}\label{snogo}

We start by recalling how to build on the geodesic approximation to classical dynamics of GLSMs (see Section~\ref{sec:adiabatic})
to study the quantization of certain supersymmetric extensions of these models. This will follow the basic semi-classical scheme proposed in~\cite{BokRomP,RomWeg,BokRomWeg},
here aimed at studying ground states of the topological A-twist~\cite{BapTSM} in terms of supersymmetric quantum mechanics on the moduli spaces $\mathcal{M}(n,r,d)$.
An immediate consequence of Theorem~\ref{lefg} in this context (for $n=r$) is given in Section~\ref{sec.nononab}.

\subsection{A semi-classical quantization scheme for supersymmetric GLSMs}\label{sec:semiclass}

We have seen that setting $\xi\,=\,1$ renders the GLSM defined by the action (\ref{GLSM}) self-dual, in the sense that 
one can describe the static stable field configurations as solutions to the system (\ref{vortex}) via the 
Bogomol'ny\u\i\ trick. Another important feature of this critical value of $\xi$ is that it allows for the 
construction of $\mathcal{N}\,=\,2$ supersymmetric versions of the GLSMs (see~\cite{BapTSM}), provided that one 
supplements the bosonic fields $(A,u)$ by other fields so as to fill out vector and chiral supermultiplets. More 
concretely: to implement such an extension for a local Euclidean model, one would need to add terms to the action 
that also involve an adjoint scalar field $\sigma$, four fermionic fields $\psi_\pm, \lambda_\pm$ and two scalar 
auxiliary fields $F,D$.

However, in order to have a supersymmetric version of the model defined on an arbitrary surface $\Sigma$, one needs 
to perform a topological twist~\cite{Wit-QFS}. The twist we will be interested in is the A-twist that uses the 
vector (global) circle ${\rm R}$-symmetry. The Lagrangian and spectrum of the corresponding twisted version of the 
two-dimensional GLSM is described in Section~3.1 of reference~\cite{BapTSM}. In Section~3.3 of the same reference, 
it is argued that the path integrals of the twisted model localize to the moduli space of vortices defined 
in~(\ref{mo}), and that its observables can be interpreted in terms of the Hamiltonian Gromov--Witten invariants 
of~\cite{CGMS}.

In view of this localisation phenomenon, one should hope to understand the ``BPS sector'' of the quantized twisted 
supersymmetric GLSMs via canonical quantization of the truncated phase spaces ${\rm 
T}^*\mathcal{M}_{\Sigma}(n,r,d)$. It is well known~\cite{HV} that one-dimensional sigma-models onto a K\"ahler 
target manifold such as our moduli spaces $\mathcal{M}_\Sigma(n,r,d)$ (whenever smooth) admit extensions with 
$\mathcal{N}\,=\,(2,2)$ supersymmetry, and can thus accommodate even the full amount of local supersymmetry present in 
the classical two-dimensional theory. This ``semi-classical'' regime should capture the physics at low energies for 
each $d\,>\,0$, and in particular the structure of the ground states, which can be described using the framework of 
supersymmetric quantum mechanics~\cite{WitSMT,HV}. We shall assume from now on that $n\,=\,r$ is fixed as in 
(\ref{mosimp}).

According to the original proposal of Witten~\cite{WitSMT}, the ground states in the effective supersymmetric 
quantum mechanics should correspond to harmonic waveforms on each moduli space $\mathcal{M}_\Sigma \,=\,
\mathcal{M}_\Sigma(n,n,d)$, with respect to its natural K\"ahler metric $g_{L^2}$. The supersymmetric parity of such 
states will be governed by their degree as differential forms reduced mod 2. However, multiparticle quantum states, 
corresponding to moduli spaces at degrees $d>1$, may also admit an interpretation in terms of individual solitonic 
particles. This leads to an expectation that the Hilbert spaces obtained from the quantization of each component 
$\mathcal{M}_\Sigma(n,n,d)$ with $d\,>\, 1$ might split nontrivially into sums of tensor products of elementary Hilbert 
spaces corresponding to constituent particles. This phenomenon is illustrated in reference~\cite{RomWeg} in the 
context of the simplest possible gauged sigma-model with nonlinear target (the round two-sphere with usual circle 
action).

In the picture we are proposing, it is natural to extend the semi-classical approximation by allowing nontrivial 
holonomies of the waveforms in supersymmetric quantum mechanics. This grants to the quantum particles the 
possibility of braiding with nontrivial anyonic phases, in analogy with the Aharonov--Bohm effect. 
Following~\cite{Wit-QFS}, one could thus advocate that the waveforms be valued in local systems over the moduli 
space (constructed from representations of its fundamental group); or equivalently~\cite{BokRomP}, one performs the 
quantization of a cover of the moduli space $\mathcal{M}_\Sigma$ where the relevant local systems trivialize --- of 
course, this will always be the case for the universal cover $\widetilde{\mathcal M}_\Sigma$.

Another extension~\cite{BokRomP,BokRomWeg} is to allow for wavepackets of waveforms, using linear combinations over the representation
variety of $\pi_1(\mathcal{M}_\Sigma)$ rather than a fixed representation. So we are lead to taking as quantum Hilbert space the $L^2$-completion of the space of forms with compact support~\cite{AtiEO}
$$
L^2\Omega_c^*\left( \widetilde{\mathcal{M}}_\Sigma ;\CC \right) \,\cong\,
 \Omega^*\left( \mathcal{M}_\Sigma;\CC \right) \otimes \ell^2\left( \pi_1\, \mathcal{M}_\Sigma \right).
$$
The ground states are to be sought among the harmonic forms with respect to the metric $g_\Sigma$, but we expect the space of harmonic forms to be infinitely generated in crucial examples --- this follows from
the fibration (\ref{Phi}) and results in \cite[Sec.~4]{BokRomC}. In order to count ground states meaningfully, one needs to resort to renormalized dimensions in the sense of
Murray--von Neumann~\cite{MurvNeu,BokRomP}, and then such counting corresponds to the computation of analytic $L^2$-Betti numbers~\cite{Lue} of the covers. For $d=1$ the situation is rather
simple, and it was dealt with in Theorem~14 of reference~\cite{BokRomP} (using local systems of rank one, and assuming that the genus $g$ of $\Sigma$ is positive). It was shown that the ground states of single solitons are fermionic, and can be understood effectively in terms
of ``Pochhammer states'' constructed from certain pair-of-pants decompositions of
the surface $\Sigma$. 

\subsection{On the realization of nonabelionic statistics} \label{sec.nononab}

It would be desirable to calculate the $L^2$-Betti numbers of the moduli spaces $\mathcal{M}_\Sigma$ beyond the 
$d\,=\,1$ case, and draw conclusions about the spectrum of {\em multiplarticle} ground states of the GLSMs. In 
particular, one would like to classify the constituent particles according to their statistics (equivalently, 
understand how they braid on the surface $\Sigma$). This is to be
contrasted with ordinary nonrelativistic 
quantum-mechanics, in which the quantum Hilbert space of one particle is constructed first, and then multiparticle 
states are obtained a posteriori, implementing by hand bosonic/fermionic/anyonic statistics. In our context, the 
statistics of the particles are imposed by the geometry and topology of the moduli spaces.

For the GSLMs studied in this paper, all these tasks are difficult, and require detailed information about the 
structure of the fibration (\ref{Phi}) --- not only a description of the fibres, as given in Section~\ref{neqr}, but 
also how the different strata (corresponding to partitions of $d$) glue together. However, one has the following 
immediate corollary of Theorem~\ref{lefg}:

\begin{corollary}\label{nogo}
An irreducible local system over a multiparticle moduli space (\ref{mosimp}) with $d\,>\,1$ must have rank
one, and the semi-classical quantization
scheme discussed in Section~\ref{sec:semiclass} rules out constituent particles with nonabelionic statistics.
\end{corollary}

\begin{proof}
By Corollary~\ref{pi1M}, we have $\pi_1(\mathcal{M}_\Sigma(n,n,d))\,\cong\, \ZZ^{\oplus 2g}$, where $g$ is the genus
of the surface $\Sigma$. In particular, all its irreducible representations are one-dimensional.
\end{proof}

The situation here is analogous to the abelian linear model $r\,=\,n\,=\,1$, discussed briefly in~\cite{BokRomP}. We
would like to point out that there exist plenty of {\em abelian} gauged sigma-models constructed from nonlinear target actions,
for which the associated vortex moduli spaces turn out to have {\em nonabelian} fundamental groups, see~\cite{BokRomB}. Those models may already support
nonabelions in quantum multiparticle states, in contrast to the GLSMs we considered in this paper.

\section*{Acknowledgements}

This work draws on discussions held at the occasion of the program ``The Geometry, Topology and Physics of Moduli 
Spaces of Higgs Bundles'' at the Institute for Mathematical Sciences, National University of Singapore; the authors 
would like to thank the organizers and NUS for hospitality. The first-named author is supported by a J. C. Bose
Fellowship.

\end{document}